\documentclass[final,5p,times,twocolumn]{elsarticle}

\usepackage{amsmath,empheq}
\usepackage{amsfonts}
\usepackage{amsthm}
\usepackage{amssymb}
\usepackage{graphicx}
\usepackage{hyperref}
\usepackage{xcolor}
\usepackage{subfigure}
\usepackage[english]{babel}

\hypersetup{
colorlinks,
linkcolor={red},
citecolor={blue},
urlcolor={blue}
}
\usepackage{color}
	\newcommand{\ncd}{\newcommand}
	\ncd{\mrm}{\mathrm}
	\ncd{\red}{\color{red}}
	\ncd{\blue}{\color{blue}}
	\ncd{\beq}{\begin{equation}}
	\ncd{\eeq}{\end{equation}}

	\newtheorem{definition}{Definition}[section]
	\newtheorem{prop}{Proposition}[section]

	\def\d{{\rm d}}

	\newcommand{\figref}[1]{Fig.~\ref{#1}}
	\newcommand{\figsref}[1]{Figs.~\ref{#1}}
	\newcommand{\eqeqref}[1]{Eq.~\eqref{#1}}
	\newcommand{\eqsref}[1]{Eqs.~\eqref{#1}}
	\newcommand{\secref}[1]{Section~\ref{#1}}

\begin{document}

\begin{frontmatter}
\title{The zeroth law in quasi-homogeneous thermodynamics and black holes}

\author[iimas]{Alessandro Bravetti\fnref{corr}}
\ead{alessandro.bravetti@iimas.unam.mx}
\author[icn]{Christine Gruber}
\ead{christine.gruber@correo.nucleares.unam.mx}
\author[conacyt]{Cesar S. Lopez-Monsalvo}
\ead{cslopezmo@conacyt.mx}
\author[icn]{Francisco Nettel}
\ead{fnettel@nucleares.unam.mx}

\fntext[corr]{Corresponding author}

\address[iimas]{Instituto de Investigaciones en Matem\'aticas Aplicadas y en Sistemas,\\
		Universidad Nacional Aut\'onoma de M\'exico, Ciudad Universitaria, Ciudad de M\'exico 
		04510, Mexico A.P. 70-543, 04510 Ciudad de M\'exico, M\'exico}
\address[icn]{Instituto de Ciencias Nucleares, Universidad Nacional Aut\'onoma de M\'exico,\\	
		A.P. 70-543, 04510 Ciudad de M\'exico, M\'exico}
\address[conacyt]{Conacyt-Universidad Aut\'onoma Metropolitana Azcapotzalco
		Avenida San Pablo Xalpa 180, Azcapotzalco, Reynosa Tamaulipas, 02200 Ciudad de 
		M\'exico, M\'exico}

\begin{abstract}
Motivated by black holes thermodynamics, 
we consider the zeroth law of thermodynamics for systems whose entropy is a quasi-homogeneous function 
of the extensive variables. We show that the generalized Gibbs-Duhem identity and the Maxwell construction 
for phase coexistence based on the standard zeroth law are incompatible in this case. We argue that 
the generalized Gibbs-Duhem identity suggests a revision of the zeroth law which in turns permits to reconsider 
Maxwell's construction in analogy with the standard case. The physical feasibility of our proposal is considered in the particular case of black holes.\end{abstract}


\end{frontmatter}

\section{Introduction}

The thermodynamics of black holes contains several peculiarities in contrast to standard thermodynamics. 
One example is the different scaling behaviour when rescaling the thermodynamic variables. This can be directly verified noting that the entropy of a black hole is --  in general -- a quasi-homogeneous function of the extensive thermodynamic quantities describing the system~\cite{dolan2015black}, 
and its scaling behaviour is dictated by the Smarr relation.
Such systems are generically called \emph{quasi-homogeneous}.
As a consequence, it is usually recognized that using the formalism of {homogeneous} thermodynamics in the case of black 
holes is not fully justified and that a modification of the thermodynamic laws for systems with 
quasi-homogeneous entropy is called for~\cite{dolan2015black}.

{It has been established that in systems  where } entropy and energy are not additive the standard way to define equilibrium 
has to be adjusted  and, in such case,  the thermodynamic temperature  may not be the correct parameter to be equated at 
equilibrium~\cite{oppenheim2003thermodynamics,abe2001general,
ramirez2008violation,ramirez2008systems,biro2011zeroth,lenzi2012extensive,haggard2013death,velazquez2016remarks}.
In spite of this, {it has been repeatedly argued in favour of}  
the existence of first order phase transitions -- i.e.,~coexistence processes -- {within the framework of black hole thermodynamics}. 
{Such arguments are} based on the analogy with the van der Waals (vdW) phase diagram and use the Maxwell equal 
area law to find the coexistence curve \emph{as if} the system was homogeneous (see 
e.g.~\cite{dolan2015black,2012Dola,kubizvnak2012p,spallucci2013maxwell,spallucci2013maxwell?s,dolan2014vacuum,lan2015note,wei2015clapeyron,wei2015insight,mo2015coexistence,kubizvnak2017black} 
and the references therein). 

In this work we 
consider systems whose entropy is a quasi-homogeneous function of the extensive variables 
and show that Maxwell's equal area law -- based on the definition of thermodynamic equilibrium 
for homogeneous systems (cf.~\cite{callen2006thermodynamics} and the discussion in Section 4.3 
in~\cite{Bravetti2015377}) --  is inconsistent with the generalized Gibbs-Duhem (GGD) identity that must 
hold in such cases~\cite{belgiorno2003quasi,belgiorno2011general}. 
{We show that this situation can be remedied introducing a new set of the variables 
defining equilibrium. Based on these generalized variables, we propose a definition of thermodynamic equilibrium {originating} from the 
GGD identity and we demonstrate that such revision is essential {in} Maxwell's construction 
for phase coexistence. It is worth mentioning that our \emph{generalized zeroth law} reduces to the standard definition for homogeneous systems of degree one.}

To illustrate our proposal we discuss two relevant cases: {on the one hand}, we show that for the Schwarzschild 
black hole {the \emph{new} temperature characterizing equilibrium is constant, i.e. it does not depend on its mass $M$. This} coincides (up to a constant factor) with 
the result in~\cite{czinner2015black}, 
where such parameter is obtained using a generalized zeroth law for non-extensive statistical mechanics developed 
in~\cite{biro2011zeroth}. This proves that, at least in the Schwarzschild case, there is a consistency between 
different approaches. On the other hand, we consider the first order phase transition in the Kerr--Anti de Sitter (Kerr--AdS) family of black holes and show that the Maxwell construction 
as applied in the literature leads to a violation of the GGD. Using 
the new generalized intensive parameters and according to our definition of thermodynamic equilibrium, such transition seems to 
disappear. Given the importance of this example in the context of the AdS/CFT correspondence, we believe that this 
can be relevant for future investigations.

This paper is structured as follows. In \secref{sec:mathincons} we review the thermodynamics of 
quasi-homogeneous systems as developed in~\cite{belgiorno2003quasi,belgiorno2011general}. In~\secref{sec:GZL} we 
point out the aforementioned mathematical inconsistency between Maxwell's construction based on the 
standard zeroth law of thermodynamics and the Gibbs-Duhem relation in the case of quasi-homogeneous entropy, 
and continue by proposing a generalized form of the zeroth law, 
which is consistent with the corresponding GGD relation. To illustrate the new form of the zeroth law, we consider the examples of 
Schwarzschild and Kerr--AdS black holes in \secref{sec:ex}, before we conclude in \secref{sec:conc}. 
Throughout this work we use Planck units, in which $c=G=\hbar=k_{\rm B}=1$.

\section{Quasi-homogeneous thermodynamics}
\label{sec:mathincons}
In this section we briefly review some results of the thermodynamics of quasi-homogeneous systems
obtained in~\cite{belgiorno2003quasi,belgiorno2011general}. 
Let us start by recalling some definitions. 
Unless otherwise stated, we will not use Einstein's sum convention.

\begin{definition}[Quasi-homogeneous function] \label{def:qhf}
Let $r,\lambda\in \mathbb{R}$, $\lambda\neq 0$ and ${\bf \beta} = \left(\beta_1,...,\beta_n\right)\in \mathbb{R}^n$. 
A function $w$ of a set of variables $\left\{q^i\right\}_{i=1}^n$ is said to be \emph{quasi-homogeneous of degree 
$r$ and type ${\bf \beta}$} if
	\beq \label{eq:Gquasihom}
	  w (\lambda^{\beta_1} q^1,..., \lambda^{\beta_n} q^n) = \lambda^{r} w(q^1,...,q^n) \,.
	\eeq
\end{definition}
The particular case where $\beta_i=1$ for every value of $i$ yields the standard scaling relation of homogeneous 
functions of degree $r$, i.e.
	\beq
	\label{eq.whom}
	w(\lambda q^1,...,\lambda q^n) = \lambda^r w(q^1,...,q^n) \,.
	\eeq
In the following we will use $S$ instead of $w$, because the function of interest in thermodynamics is the entropy. 
The variables $\left\{q^i\right\}_{i=1}^n$ are the extensive variables of the system, such as internal energy $U$, 
volume $V$ or number of particles $N$.
{In standard thermodynamics} of extensive systems the entropy is a homogeneous function of degree 
one of the extensive variables, i.e., 
	\beq\label{homS}
	S(\lambda U,\lambda V,\lambda N) = \lambda S(U,V,N) \,, 
	\eeq
while in black holes thermodynamics the entropy is a quasi-homogeneous function as in Def.\,\ref{def:qhf}.

\begin{prop}
Let $S=S(q^1,...,q^n)$ be a quasi-homogeneous function of  degree $r$ and type $\beta$. Then, the conjugate 
variables to the $q^i$, defined by   
	\beq
		\label{def.p0}
	p_i\left(q^j\right) \equiv \frac{\partial}{\partial q^i} S\left(q^j\right) \,,
	\eeq 
are quasi-homogeneous functions of degree $r-\beta_i$ for every value of $i$.
\end{prop}
\begin{proof}
	
	\begin{align}
	\label{eq.p1}
		p_i\left(\lambda^{\beta_j} q^j\right) = \frac{\partial}{\partial (\lambda^{\beta_i} q^i)} 
			S\left(\lambda^{\beta_j} q^j\right)
			& = \frac{1}{\lambda^{\beta_i}} \frac{\partial}{\partial q^i}\big[\lambda^r S\left(q^j\right)\big] 
					\nonumber \\
			& = \lambda^{r-\beta_i}\frac{\partial}{\partial q^i} S(q^j) \,.
	\end{align} 
Therefore
	\beq
	\label{hombeta_a-1}
	p_i\left( \lambda^{\beta_j} q^j \right) =\lambda^{r-\beta_i} p_i (q^j) \,.
	\eeq
\end{proof}

Note that if $S$ is homogeneous of degree $r=1$  [cf.~equation \eqref{homS} above], 
then the conjugate variables $p_i$ are homogeneous functions of degree 0, i.e.~$p_i(\lambda q^j) = p_i(q^j)$, 
i.e., they do not change when the system is re-scaled. Only in this case, the conjugate variables 
are \emph{intensive} and we recover the usual thermodynamic quantities, e.g.~$1/T$, $p/T$, 
$\mu/T$. In all other cases we shall refer to the conjugate variables $p_i$ as the \emph{would-be intensive} 
quantities, as in~\cite{belgiorno2003quasi,belgiorno2011general}. 

\begin{prop}[Euler's Theorem]
\label{teo.eu}
Let $S=S(q^1,...,q^n)$ be a quasi-homogeneous function of  degree $r$ and type $\beta$. Then
	\beq
	\label{eu.teo}
	r S(q^j) = \sum_{i=1}^n \beta_i \big[q^i p_i(q^j) \big] \,.
	\eeq
\end{prop}
\begin{proof}

Consider the derivative of $S(\lambda^{\beta_j} q^j)$ with respect to the scaling parameter $\lambda$. 
On the one hand, since $S$ is a quasi-homogeneous function of degree $r$ and type $\beta$, we have
	\begin{align}
	\label{teo.eu1}
		\frac{\partial}{\partial \lambda} S(\lambda^{\beta_j} q^j)	
		=\frac{\partial}{\partial \lambda} \big[\lambda^r S(q^j) \big] 
		 = r \lambda^{r-1} S(q^j) \,.
	\end{align}
On the other hand, a direct calculation yields
	\begin{align}
	\label{teo.eu2}
	  \frac{\partial}{\partial \lambda} S(\lambda^{\beta_j} q^j) 
		& = \sum_{i=1}^n \frac{\partial S(\lambda^{\beta_j} q^j)}{\partial (\lambda^{\beta_i} q^i)} 
			\frac{\partial (\lambda^{\beta_i} q^i)}{\partial \lambda} \nonumber \\
	  & = \sum_{i=1}^n \frac{\partial S(\lambda^{\beta_j} q^j)}{\partial (\lambda^{\beta_i} q^i)} 
			\left(\beta_i \lambda^{\beta_i-1} q^i \right) \nonumber\\
		& =  \sum_{i=1}^n \left(\beta_i \lambda^{r-1} q^i \right)  p_i(q^j),
	\end{align} 
where the last equality follows from Def.\,\eqref{def.p0} and \eqsref{eq.p1} and \eqref{hombeta_a-1}. 
Thus, combining the results of \eqref{teo.eu1} and \eqref{teo.eu2}, \eqeqref{eu.teo} is obtained. 
\end{proof} 

In standard thermodynamics the above result reduces to the well-known identity for the entropy, 
	\beq\label{eulers1}
		S=\frac{1}{T}U-\frac{p}{T}V+\frac{\mu}{T}N \,.
	\eeq

With Proposition~\ref{teo.eu}, we can write a GGD relation for the case of quasi-homogeneous thermodynamic systems.

\begin{prop}[Generalized Gibbs-Duhem identity]
Let $S(q^{1},\dots,q^{n})$ be a quasi-homogeneous function of  degree $r$ and type $\beta$ 
and let $\left\{p_i\right\}_{i=1}^n$ be the set of conjugate variables [cf.~equation~\eqref{def.p0}]. 
Then, 
	\beq\label{eq:GD-QH}
		\sum_{i=1}^n  \big[ \left(\beta_i - r \right) p_i(q^j) \d q^i+ \beta_i q^i \d p_i(q^j) \big] =0 \,.
	\eeq
\end{prop}
\begin{proof}
Since $S$ satisfies the hypothesis of Proposition \ref{teo.eu}, let us consider the differential of 
\eqref{eu.teo}, namely 
	\beq
	r \d S(q^j) = \sum_{i=1}^n\beta_i \d\big[q^i p_i(q^j) \big] \,.
	\eeq
The left hand side is simply
	\beq
	\label{gd.1}
		r \d S =r \sum_{i=1}^n \frac{\partial}{\partial q^i} S(q^j) \d q^i= r \sum_{i=1}^n p_i(q^j) \d q^i \,,
	\eeq
whereas the right hand side yields
	\beq
	\label{gd.2}
		\sum_{i=1}^n\beta_i \d\big[q^i p_i(q^j) \big] = 
		\sum_{i=1}^n \beta_i \big[q^i \d p_i(q^j) + p_i(q^j) \d q^i \big] \,.
	\eeq
Subtracting \eqref{gd.1} from \eqref{gd.2} and collecting the $\beta_i$ produces the desired result. 
\end{proof}

In the case where $S$ is homogeneous of degree $r$, equation 
\eqref{eq:GD-QH} reduces to 
	\beq\label{G-D2}
		(1-r) \sum_{i=1}^n p_i(q^j) \d q^i +  \sum_{i=1}^n q^i \d p_i(q^j) = 0 \,.
	\eeq
From this result it follows that in standard thermodynamics (with $r=1$), using the appropriate identifications of the variables, one obtains the Gibbs-Duhem relation
	\beq\label{GDstandard}
		U\d\left(\frac{1}{T}\right)-V\d\left(\frac{p}{T}\right)+N\d\left(\frac{\mu}{T}\right)=0\,,
	\eeq
which is a mathematical identity stating that the intensive quantities are not all independent 
in equilibrium~\cite{callen2006thermodynamics}.

\section{A mathematical inconsistency and its possible resolution}
\label{sec:GZL}
In this section we prove the mathematical inconsistency between the usual zeroth law 
of thermodynamics, the standard Maxwell construction for coexistence between different phases 
and the GGD identity, and provide a possible resolution through a 
redefinition of the equilibrium parameters.

We start from the crucial fact that in ordinary thermodynamics the Gibbs-Duhem identity 
\eqref{GDstandard} is mathematically consistent with Maxwell's law for phase coexistence.
Here, one considers a single system splitting into two different phases remaining at equilibrium, 
i.e.~sharing the same values of their intensive quantities, while the entropy and volume of the system 
change, causing a discontinuity in the extensive quantities and thus giving rise to a \emph{first order 
phase transition. }
Clearly in this case the definition of equilibrium between the phases in terms of equal values 
of the conjugate (intensive) quantities is consistent with~\eqref{GDstandard}.

From the above discussion on the role of the intensive variables in Maxwell's construction and its 
consistency with the Gibbs-Duhem relation \eqref{GDstandard}, it is evident why such consistency is 
lost in the case of quasi-homogeneous systems, where equation~\eqref{eq:GD-QH} holds. Indeed for the 
two phases to be at equilibrium, the zeroth law would predict that no change in any of the would-be intensive  
variables $p_{i}$ would happen, i.e.,~$\d p_{i}=0$ for all $i$. This implies that the second term 
in~\eqref{eq:GD-QH} vanishes identically. However, in general the first term in~\eqref{eq:GD-QH} is different 
from zero, thus {leading to} an inconsistency.
For instance, in the case of a homogeneous entropy of degree $r$, it follows from the first law 
$\d S = \sum_{i=1}^n p_i \d q^i$ that the first term of \eqref{G-D2} is proportional to the change in 
the entropy during the transition, and hence to the latent heat, which cannot be zero in a first order 
phase transition. 

This inconsistency leads to the two following possibilities: either one gives up the standard formulation 
of phase coexistence expressed by the Maxwell construction (at least in its usual form), or one has to {re-define} 
the conditions for equilibrium, i.e.,~the zeroth law. Due to the many indications arising from different 
perspectives pointing to the fact that the zeroth law needs to be revisited for systems with non-additive 
entropy and energy relations (see e.g.~\cite{oppenheim2003thermodynamics,abe2001general,
ramirez2008violation,ramirez2008systems,biro2011zeroth,lenzi2012extensive,haggard2013death,velazquez2016remarks}), 
we opt for the latter route.

From the analysis of the homogeneity of the first derivatives of $S$ -- see \eqref{hombeta_a-1} --  
let us propose the following
\begin{definition}[Generalized intensive variables] \label{def.GIV}
Let $S(q^{1},\dots,q^{n})$ be a quasi-homogeneous function of  degree $r$ and type $\beta$ and let
$\left\{p_i\right\}_{1}^{n}$ be the set of conjugate variables. Assume that $\beta_i\neq 0$ for every $i$. 
The quantities 
	\beq \label{eq:GenIntQuantQH}
	  {\tilde{p}_i(q^j)} \equiv \left[\left(q^i\right)^{\beta_i-r}\right]^{1/\beta_i} p_i (q^j) \,.
	\eeq
are called the \emph{generalized intensive variables}.
\end{definition}

Indeed, these variables reduce to~\eqref{def.p0} when $S$ is homogeneous of degree 1. Moreover, one can easily 
prove the following
\begin{prop}\label{PropHomog}
The generalized intensive variables \eqref{eq:GenIntQuantQH} are quasi-homogeneous functions of degree 0.
\end{prop}
\begin{proof}
\begin{align}
	{\tilde{p}_i(\lambda^{\beta_j} q^j)} 
	&= \left[\left(\lambda^{\beta_i} q^i\right)^{\beta_i-r}\right]^{1/\beta_i} p_i (\lambda^{\beta_j}q^j)\nonumber\\
	&=\lambda^{\beta_i - r} \left[\left(q^i\right)^{\beta_i-r}\right]^{1/\beta_i} 
		\big[\lambda^{r-\beta_i} p_i(q^j)\big]	\nonumber\\
	&=\left[\left(q^i\right)^{\beta_i-r}\right]^{1/\beta_i} p_i (q^j) = {\tilde{p}_i(q^j)} \,.
\end{align}
\end{proof}

This is a desirable property for quantities defining a notion of equilibrium {as they remain invariant under a scaling of the system. 
Note that these generalized variables {could have been} inferred from Eq.~(75) in~\cite{belgiorno2003quasi}. 
However, in that work they were not singled out nor were advocated as the {correct} ones to describe equilibrium. 

Using the generalized intensive variables~\eqref{eq:GenIntQuantQH}, 
we can re-write the GGD identity~\eqref{eq:GD-QH} as in~\cite{belgiorno2003quasi}}:
\begin{prop}\label{prop.GGD} 
Let $S(q^{1},\dots,q^{n})$ be a quasi-homogeneous function of  degree $r$ and type $\beta$ and let 
$\left\{{\tilde{p}_i} \right\}_{i=1}^n$ be the set of generalized intensive variables. Then, 
	\beq\label{GDnew}
	\sum_{i=1}^n \beta_i \left(q^i\right)^{r/\beta_i}  \d \tilde{p}_i(q^j) = 0\,.
	\eeq
\end{prop}
\begin{proof}
From \eqeqref{eq:GenIntQuantQH} we have 
	\beq
	  p_{i}(q^j) = {\tilde{p}}_i(q^{j}) \, {\left(q^i\right)}^{r/\beta_i -1} \,,
	\eeq
and we can thus rewrite the identity \eqref{eq:GD-QH} in terms of the $\tilde{p}_i(q^j)$ as 
	\beq\label{G-D3quasi}
	\begin{split} 
		\sum_{i=1}^n \beta_i \Big[ \left(1 - \frac{r}{\beta_i} \right) {\tilde{p}}_i \,{\left(q^i\right)}^{r/\beta_i - 1}\, 
			\d q^i 
		+ q^i \d\left( {\tilde{p}}_i \,{\left(q^i\right)}^{r/\beta_i - 1} \right) \Big] = 0\,.
	\end{split}
	\eeq
By explicit calculation of the second term, we can rewrite the above identity as  
	\begin{eqnarray}
	  0 &=& \sum_{i=1}^n \beta_i \Bigg[ \left( 1 - \frac{r}{\beta_i} \right) {\tilde{p}}_i \, 
			{\left(q^i\right)}^{r/\beta_i - 1} \, \d q^i  \nonumber \\
	  &+& q^i \Big[ {\left(q^i\right)}^{r/\beta_i-1}\,\d {\tilde{p}}_i+ \left( \frac{r}{\beta_i} -1 \right) 
			{\left(q^i\right)}^{r/\beta_i-2} \,{\tilde{p}}_i \,\d q^i \Big] \Bigg] \nonumber\\
	  &=& \sum_{i=1}^n \beta_i\, {\left(q^i\right)}^{r/\beta_i} \d \tilde{p}_i \,. \label{G-D4quasi}
	\end{eqnarray}
 \end{proof}

Note that the GGD identity~\eqref{eq:GD-QH} only 
establishes the existence of a relation between the would-be 
intensive and the would-be extensive variables, without fixing the 
values of the generalized intensive variables uniquely. In this 
sense our choice of the generalized intensive variables~\eqref{eq:GenIntQuantQH} is not the only one possible. 
However, it is motivated by the following considerations. Firstly,~\eqref{eq:GenIntQuantQH} reduce to~\eqref{def.p0} when the entropy is homogeneous of degree $1$. 
Moreover, these quantities are quasi-homogeneous functions of degree $0$, thus being true intensive variables 
(under the appropriate re-scalings of the extensive ones). 
Finally, as stated in Proposition~\ref{prop.GGD}, using these variables the GGD identity takes the same form as the standard one (cf.~\cite{belgiorno2003quasi}).
Indeed, Propositions~\ref{PropHomog} and~\ref{prop.GGD} suggest the following modification of the notion of 
thermodynamic equilibrium: 

\begin{definition}[Thermodynamic Equilibrium] \label{def.TE}
Two systems whose entropy is a quasi-homogeneous function of the same degree {and type} are in thermodynamic equilibrium 
with each other if and only if they have the same values of the ${\tilde{p}}_i(q^j)$.
\end{definition}

This is the \emph{generalized zeroth law of thermodynamics} that we propose for any quasi-homogeneous system. 
Note that Def.\,\ref{def.TE} is mathematically consistent with the identity~\eqref{eq:GD-QH} 
 -- cf.~\eqref{GDnew} -- even when considering processes of coexistence as in the case of the usual 
Maxwell equal area law. 

Let us remark that with our prescription one can consider the example of a process of coexistence among 
different phases at equilibrium without any incongruence, as long as the definition of equilibrium is given by 
equating the quantities in~\eqref{eq:GenIntQuantQH}. Note also that our simple redefinition gives a general 
prediction about the quantities that have to be constant at equilibrium. 

In the next section we consider examples from black holes thermodynamics and show that for the Schwarzschild 
black hole our redefinition of the equilibrium condition {yields a constant generalized temperature. This result} coincides with a different instance of the generalized 
zeroth law of thermodynamics resulting from non-extensive statistical mechanics~\cite{czinner2015black}. {As a more relevant consequence we will also show that}
for the Kerr--AdS black hole our construction suggests that a reconsideration of the first order phase transition 
might be in order.

\section{Quasi-homogeneous black hole thermodynamics} 
\label{sec:ex}
In this section we investigate some examples for the above ideas in the context of black hole thermodynamics. 
In principle, our generalization of the zeroth law can be applied to any black hole system, given that one can 
easily determine the degrees of homogeneity from the Smarr relation, 
	\begin{equation}\label{smarr}
	(D-3)M = (D-2)TS+(D-2)\Omega J 
	-2PV +(D-3)\Phi Q
	\end{equation}
where $D$ is the number of spacetime dimensions, $M$ is the mass of the black hole, $T$ is the Hawking 
temperature, $S$ is the entropy and the other terms are work terms depending on the black hole family in 
question~\cite{dolan2015black}. Here, we consider two in particular, namely the Schwarzschild and the Kerr--AdS 
black holes, to compare our results with previous proposals and to illustrate new features.

\subsection{Schwarzschild}
The Schwarzschild black hole is the most straightforward example, since its thermodynamics is described by 
only one extensive variable, i.e.,~its mass $M$. The entropy as a function of $M$ is 
\begin{equation}
  S(M) = 4 \pi M^2 \,,
\end{equation}
which is a homogeneous function of degree $r=2$. From this the standard temperature is derived as 
\begin{equation}\label{TSchwnormal}
  \frac{1}{T} = \frac{\partial S}{\partial M} = 8 \pi M \,.
\end{equation}
It is immediate to see that this a homogeneous function of degree $1$ with respect to $M$, and therefore not a real intensive quantity. 
With~\eqref{TSchwnormal} 
and using \eqref{eq:GenIntQuantQH}, we can obtain the generalized temperature as 
\begin{equation} \label{eq:TgenSch}
  \tilde{T} = TM = \frac{1}{8\pi} \,,
\end{equation}
i.e.,~a constant. Note that, a constant is -- trivially -- a real intensive quantity, as it does not change 
with any scaling of $M$. 
Note also that by \eqref{GDnew} the generalized intensive quantities cannot be independent. This means 
that in this case, since there is only one such generalized intensive quantity, it must be a constant. This 
fact outlines that the Schwarzschild black hole is not a proper thermodynamic system. However, it is 
interesting to see that even in this case our formalism coincides with previous approaches. Indeed, 
a similar result, i.e.,~a constant generalized temperature, has been obtained previously for the 
Schwarzschild black hole~\cite{czinner2015black} by using the generalized zeroth law derived from non-extensive 
statistical mechanics proposed in~\cite{biro2011zeroth}. In this work, the most general conditions for thermal 
equilibrium of systems with non-additive energy and entropy are established by using a method based on the 
definition of the so-called formal logarithms of these quantities. However, the same method was also applied 
in~\cite{czinner2016kerr} in the analysis of the Kerr black hole, resulting in a constant generalized temperature, 
regardless of the angular momentum, identical to the Schwarzschild case -- an indication that the result may be 
unphysical, as the authors point out themselves. 
Moreover, from our formalism a dependence of the generalized temperature on the angular momentum is to be expected. 
Finally, in~\cite{biro2013q,czinner2015RenyiKerr} the R\'enyi entropy was used as the formal logarithm of the Bekenstein-Hawking entropy. 
In this case the temperature for the Schwarzschild case depends on the mass $M$ and 
is not intensive. The connection of our proposal to these approaches and the general question of the underlying 
behaviour of the energy and entropy is thus not quite clear and might be addressed in future works.

\subsection{Kerr--AdS}
Kerr black holes in asymptotically Anti--de Sitter space are thermodynamically determined by three extensive variables, 
namely their mass $M$, angular momentum $J$ and pressure $P$, which is defined via the cosmological constant $\Lambda$ 
of the spacetime as 
\begin{equation}
 P = -\frac{\Lambda}{8 \pi} \,.
\end{equation}
The cosmological constant is usually included as a pressure into the thermodynamic description of black holes 
\cite{2009Kast,2012Dola,dolan2015black}, and thus it turns out that the internal energy of the black hole is 
\begin{equation}
  U = M - PV \,,
\end{equation}
and therefore the mass of the black hole is identified with the enthalpy
\begin{equation}
  M \equiv H = U + PV \,.
\end{equation}
For the Kerr--AdS black hole one obtains~\cite{2012Dola}
\begin{equation}
  H(S,P,J) = \frac{1}{2} \sqrt{\frac{4 \pi ^2 J^2 \left(\frac{8 P S}{3}+1\right)+\left(\frac{8 P S^2}{3}+S\right)^2}{\pi S}} \,,
\end{equation}
and from this, provided $J \neq 0$, it is possible to calculate the expression for the internal energy as 
\begin{eqnarray}
  && U(S,V,J) = \left(\frac{\pi}{S}\right)^3 \Bigg[ \left( \frac{3V}{4\pi} \right) \left\{\frac{S^2}{2\pi^2} + J^2 
    \right\} ~~~~~~~~~~ \nonumber\\
   &&~~~~~~~~~~~~~~~~~ - J^2 \left\{\left( \frac{3V}{4\pi} \right)^2 - \left( \frac{S}{\pi} \right)^3 \right\}^{1/2} \Bigg] \,.
\end{eqnarray}
For simplicity and without loss of generality, we will limit further analyses to positive angular momenta, 
i.e.,~$J>0$. 
The temperature and pressure can be easily obtained as 
	\beq
	\label{eq:temp}
	T = 	  \frac{1}{8S^4} \left[ \frac{6 \pi ^{3/2} J^2 \left(9 \pi  V^2-8 S^3\right)}{\sqrt{9 \pi  V^2-16 S^3}}  -18 \pi ^2 J^2 V-3 S^2 V \right], 
	\eeq
and
	\beq
	\label{eq:press}
	P =	 \frac{3}{8S^3} \left[ 2 \pi ^2 J \left(J-\frac{3 \sqrt{\pi } J V}{\sqrt{9 \pi  V^2-16 S^3}}\right)+S^2 \right],
	\eeq
respectively.

The case of Kerr--AdS is particularly interesting for our purposes because its equation of state, i.e.,~the relation 
$P(V,T)$ at fixed $J$, qualitatively shows the same oscillatory behaviour as a vdW fluid, which is generally taken as an 
indication of the presence of a first order phase transition, sometimes referred to as the CCK phase 
transition~\cite{dolan2014vacuum,caldarelli2000thermodynamics,altamirano2014thermodynamics}. 
To see this let us fix $J=1$ from now on and first look at \figref{fig:PoverV}, where we plot $P(V,T)$ as a function of 
$V$ for various choices of $T$, with the inlet zooming in on one of the curves to show the characteristic vdW bump. 
\begin{figure}[h!]
    \includegraphics[width=0.45\textwidth]{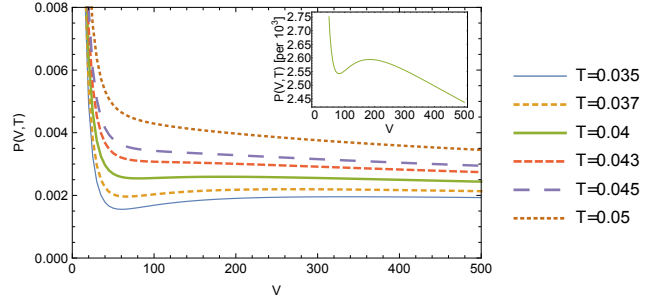}
    \caption{Equation of state $P(V,T)$ for different values of $T$ and with $J=1$.} 
		\label{fig:PoverV}
\end{figure}
The region of the bump is the area where one would apply the Maxwell equal area law in analogy to ordinary thermodynamics  
\cite{2012Dola}. A different (equivalent) way to look at such transition is by considering the graph of the Gibbs 
free energy, 
\begin{equation} \label{eq:Gtheory}
  G(T,P,J) = U - TS + PV \,.
\end{equation} 
To illustrate the multi-valued behavior of the Gibbs free energy we plot in \figsref{fig:Tconst} and \ref{fig:Pconst} 
the cuts along the lines of constant $T$ and $P$, respectively, featuring the characteristic swallowtails. 
\begin{figure}[h!]
    {\includegraphics[width=0.45\textwidth]{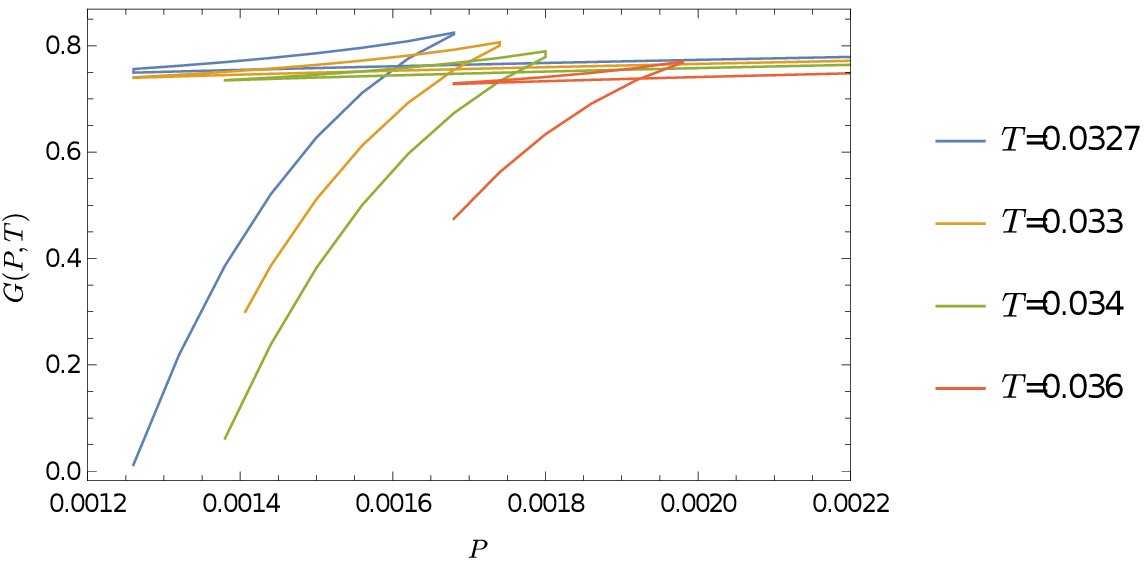}}
    \caption{Cuts of the Gibbs free energy at constant $T$.} \label{fig:Tconst}
	{\includegraphics[width=0.45\textwidth]{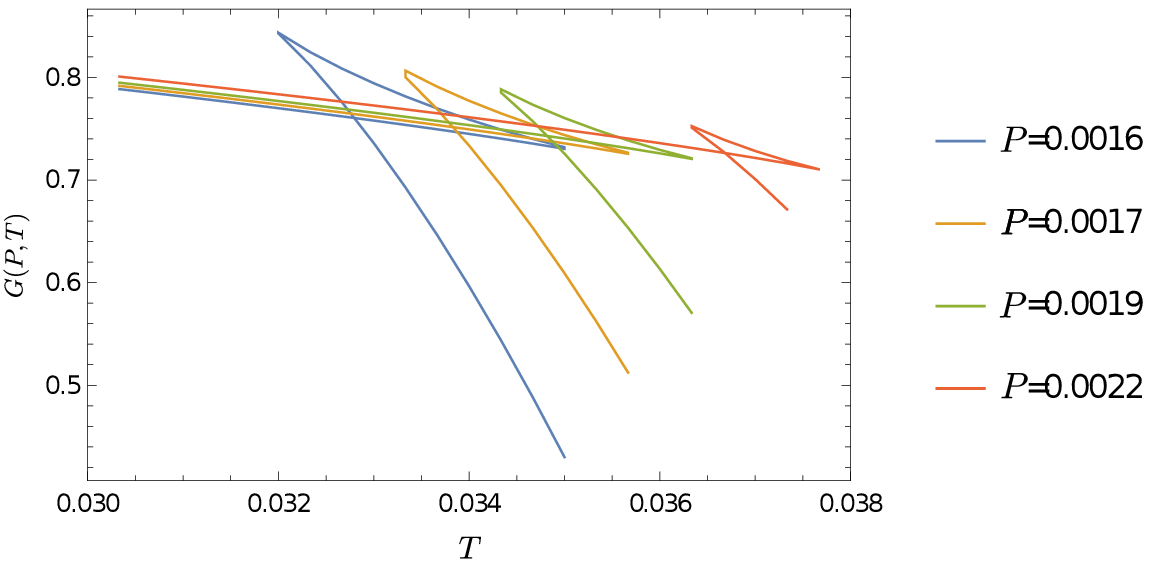}}
    \caption{Cuts of the Gibbs free energy at constant $P$.} \label{fig:Pconst}
\end{figure}

Based on the above analogy with the vdW phase diagram, it has been argued that there is a first order phase transition between
small and large Kerr--AdS black holes, for appropriate values of the temperature and pressure.
Indeed, the standard Maxwell construction can be performed, and the form of the coexistence curve can be 
calculated~\cite{dolan2014vacuum,caldarelli2000thermodynamics,altamirano2014thermodynamics,wei2016analytical}.
In the following we use this example -- which is considered to be well understood in the literature -- to claim that a revision due to the
GGD identity is called for. 

We start by showing that the standard Maxwell construction in 
this case is inconsistent with the GGD identity. 
To do so, it is more convenient to use the equation of state 
$T(S,P)$, as plotted in \figref{fig:MaxC}. 
For appropriate values of $T$ and $P$, this equation of state 
exhibits an oscillatory behavior, as for the case of $P(T,V)$ 
above (cf.~\figref{fig:PoverV}). 
	 \begin{figure}[h!]
	 \begin{center}
	    \includegraphics[width=0.45\textwidth]{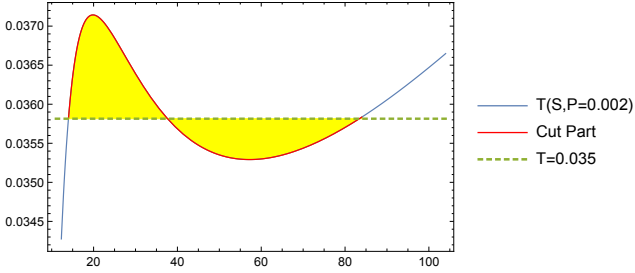}
	    \caption{Equation of state $T(S,P)$ for $J=1$ and $P=0.002$, together with the standard Maxwell construction. The two colored areas are equal.} 
		\label{fig:MaxC}
	\end{center}
	\end{figure}
The corresponding value for the transition temperature is also calculated using Maxwell's equal area law. \\
Now we proceed to verify the GGD identity for this case. 
Using the Smarr relation \eqref{smarr} applied to four spacetime 
dimensions (and $Q=0$), we have
	\begin{equation}\label{SmarrKAdS}
	  S = \frac{1}{2T} U + \frac{3}{2} \frac{P}{T} V - \frac{\Omega}{T} J \,.
	\end{equation}
From this, we can determine the degrees of homogeneity of the variables $T$, $P$ and $\Omega$ as $\beta_T = 1/2$, $\beta_P = 3/2$
and $\beta_\Omega = 1$ (cf.~Eq.~\eqref{eu.teo}). The overall degree of homogeneity of the entropy is $r=1$.
Therefore, the GGD~\eqref{eq:GD-QH} in the Kerr--AdS case reads
	\beq\label{GGDKAdS}
	-\dfrac{1}{2T}\d U+\dfrac{1}{2}\dfrac{P}{T}\d V+\dfrac{1}{2}U\d\left(\dfrac{1}{T}\right)+\dfrac{3}{2}V\d\left(\dfrac{P}{T}\right)+J\d\left(\dfrac{\Omega}{T}\right)=0 \,.
	\eeq
Note that the last three terms have an analogous form with the standard Gibbs-Duhem relation~\eqref{GDstandard} (although with different coefficients). 
Now let us use the standard Maxwell equal area law to prove an inconsistency with~\eqref{GGDKAdS}.
By the usual argument, the two coexisting phases are in equilibrium and therefore all the last three terms vanish along the coexistence process.
Thus we are left with the following expression
	\beq\label{GGDKAdS2}
	-\dfrac{1}{2T_{\rm tr}}\Delta U+\dfrac{1}{2}\dfrac{P_{\rm tr}}{T_{\rm tr}}\Delta V=0\,,
	\eeq
where $\Delta U$ and $\Delta V$ represent the jumps in these 
quantities along the coexistence line and 
$T_{\rm tr}$ and $P_{\rm tr}$ are the constant values of the temperature and pressure along the transition. 
Eq.~\eqref{GGDKAdS2} can be further simplified to
	\beq\label{GGDKAdS3}
	\Delta U-P_{\rm tr}\Delta V=T_{\rm tr}\Delta S-2P_{\rm tr}\Delta V=0\,,
	\eeq
where in the last equality we made use of the first law (with $J$ constant), that is, $\Delta U=T_{\rm tr}\Delta S-P_{\rm tr}\Delta V$.
Now it is an easy exercise to use the values 
of $T_{\rm tr}$, $P_{\rm tr}$, $\Delta S$ and $\Delta V$ calculated using the standard Maxwell construction
to show that Eq.~\eqref{GGDKAdS3} is not satisfied, i.e.,~that there is an inconsistency with the GGD identity. 
Table~\ref{table1} shows the results of these 
calculations for different values of the transition pressure 
$P_{\rm tr}$.
	In the first column we report the chosen values for the transition pressure $P_{\rm tr}$. In the second column
	we provide the corresponding transition temperature $T_{\rm tr}$, calculated using the standard Maxwell construction. In the third
	column we show that the area law is satisfied, by checking that the deviation from zero of the difference between the two areas
	in yellow in Fig.~\ref{fig:MaxC} is negligible. 
	In the last column we demonstrate that the GGD identity~\eqref{GGDKAdS3} is not satisfied, 
	by showing that the deviation from zero is large compared to that of the area law, and thus not negligible.
	\begin{table}[h!]
	\begin{center}
	\begin{tabular}{|c|c|c|c|}
	\hline 
	$P_{\rm tr}$ 		& 	$T_{\rm tr}$ 	   &  ${\rm Maxwell_{\rm Dev}}$ 	      & ${\rm GGD_{\rm Dev}}$  \\ 
	\hline 
	$1.0 \times 10^{-3}$ & $2.6 \times 10^{-2}$  & $1.8 \times 10^{-2}$ & $6.5 \times 10^{-1}$   \\ 
	\hline 
	$1.6 \times 10^{-3}$ & $3.2 \times 10^{-2}$  & $2.2 \times 10^{-2}$ & $4.6 \times 10^{-1}$  \\ 
	\hline
	$2.0 \times 10^{-3}$ & $3.5 \times 10^{-2}$  & $2.2 \times 10^{-3}$ & $3.6 \times 10^{-1}$   \\ 
	\hline 
	$2.6 \times 10^{-3}$ & $3.9 \times 10^{-2}$  & $4.5 \times 10^{-3}$ & $2.0 \times 10^{-1}$  \\ 
	\hline
	\end{tabular}
	\caption{Values of $T_{\rm tr}$, ${\rm Maxwell_{\rm Dev}}$ 	 and ${\rm GGD_{\rm Dev}}$ obtained numerically from Maxwell's equal area law for different choices of $P_{\rm tr}$. More details in the text.}
	\label{table1}
	\end{center}
	\end{table}

Since the analysis of the phase transition in terms of the standard 
definition of thermodynamic equilibrium leads to an inconsistency 
with the GGD identity, we now reconsider the phase transition in 
terms of the generalized intensive quantities defined in~\eqref{eq:GenIntQuantQH}.
From~\eqeqref{eq:GenIntQuantQH} and~\eqref{SmarrKAdS}, we can infer the generalized intensive variables responsible for equilibrium as
\begin{equation}
  \frac{1}{\tilde{T}} = \frac{1}{T U} \quad {\rm and} \quad  \frac{\tilde{P}}{\tilde{T}} = \frac{P}{T} V^{1/3} \,.
\end{equation}
Combining the two expressions, we end up with the generalized thermodynamic equilibrium parameters 
\begin{equation} 
  \tilde{T} = TU  \quad {\rm and} \quad  \tilde{P} = P U V^{1/3} \,.
\end{equation}
By construction, these functions are quasi-homogeneous of degree $0$ and type $\beta = (1,3/2,1)$ 
with respect to the correspondingly re-scaled 
extensive variables $S$, $V$ and $J$, i.e.,~
\begin{eqnarray}
   \tilde{T}(\lambda^{1} S, \lambda^{3/2} V, \lambda^{1} J) &=& \lambda^0 \,\tilde{T}(S,V,J) \,, \\
   \tilde{P}(\lambda^{1} S, \lambda^{3/2} V, \lambda^{1} J) &=& \lambda^0 \,\tilde{P}(S,V,J) \,.
\end{eqnarray}
In terms of $S$ and $V$ (for $J=1$) these read
	\begin{align}
	\label{eq:Ttilde}
	\tilde{T}(S,V) =	&  \frac{3V\pi^{13/2}}{8S^7} \left[ -36 \left( \frac{S}{\pi} \right)^3 - 10 \left( \frac{S}{\pi} \right)^5 
     + 27 \left( \frac{V}{\pi} \right)^2 +\right. \nonumber\\
     				& \left.9 \left( \frac{S}{\pi} \right)^2 \left( \frac{V}{\pi} \right)^2 \right]+ \frac{3\pi^6}{64S^7} \sqrt{9\pi V^2 - 16S^3}\ \times \nonumber\\
				& \left[ 32 \left( \frac{S}{\pi} \right)^3 - 72 \left( \frac{V}{\pi} \right)^2 
      - 24 \left( \frac{S}{\pi} \right)^2 \left( \frac{V}{\pi} \right)^2 - 3\left( \frac{S}{\pi} \right)^4 
      \left( \frac{V}{\pi} \right)^2 \right]
      \end{align}
and
	\begin{align}
	\label{eq:Ptilde}
       \tilde{P}(S,V) = &\frac{9 \pi^{5/2} V^{4/3}}{32 S^6} \sqrt{9\pi V^2 - 16S^3} \left( 2\pi^2 + S^2 \right) \ \times \nonumber\\
     				& \left( 6\pi^2 V + 3 S^2 V - 2\pi^{3/2} \sqrt{9\pi V^2 - 16S^3} \right) 
	\end{align}

Note that in order to show the quasi-homogeneity of these functions by rescaling the extensive variables, 
it is necessary to recover the terms containing $J$, including it as an extensive variable. 
Using these expressions, we can return to the plot of the equation of state, but now in terms of the new 
variables, plotting $\tilde{P}(V,\tilde{T})$ as a function of $V$ for different choices of $\tilde{T}$. 
As can be seen in \figref{fig:PtildeoverV}, the curves are monotonously decreasing, therefore the system appears 
to be stable and there is no necessity for the Maxwell construction. 
\begin{figure}[h!]
    {\includegraphics[width=0.45\textwidth]{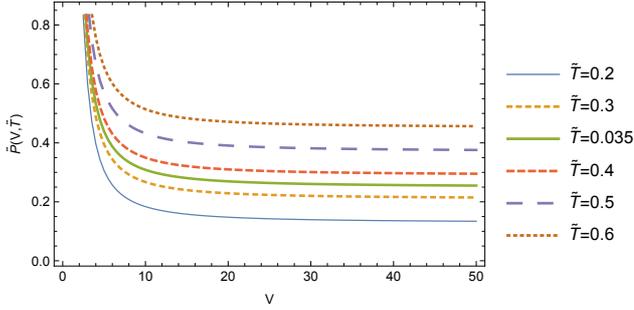}}
        \caption{Equation of state $\tilde{P}(V,\tilde{T})$ at constant $\tilde{T}$.} \label{fig:PtildeoverV}
\end{figure}
The same effect can be observed using the Gibbs free energy. We can re-express definition \eqref{eq:Gtheory} 
in terms of the new intensive variables $\tilde{T}$ and $\tilde{P}$ and calculate the function 
$G(\tilde{T},\tilde{P},J)$, inverting \eqsref{eq:Ttilde} and \eqref{eq:Ptilde} numerically. The result can be seen in 
figures \figsref{fig:Ttildeconst} and \ref{fig:Ptildeconst}, where cuts at constant $\tilde{T}$ and $\tilde{P}$ show that 
the Gibbs free energy in terms of the generalized intensive 
variables is a single-valued smooth function. \\



\begin{figure}
    {\includegraphics[width=0.45\textwidth]{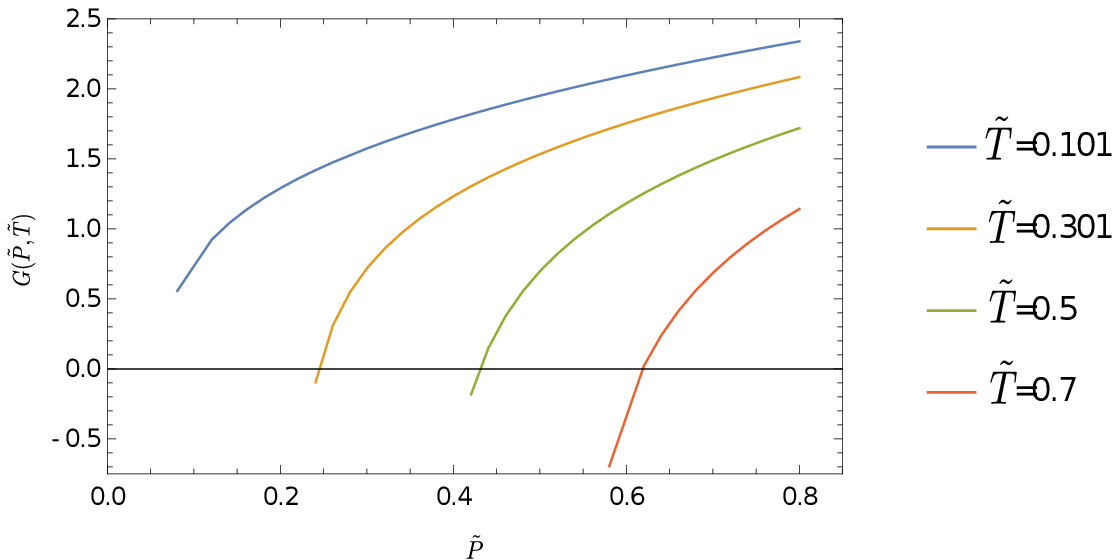}}
        \caption{Cuts  of the Gibbs free energy at constant $\tilde{T}$.} \label{fig:Ttildeconst}
	{\includegraphics[width=0.45\textwidth]{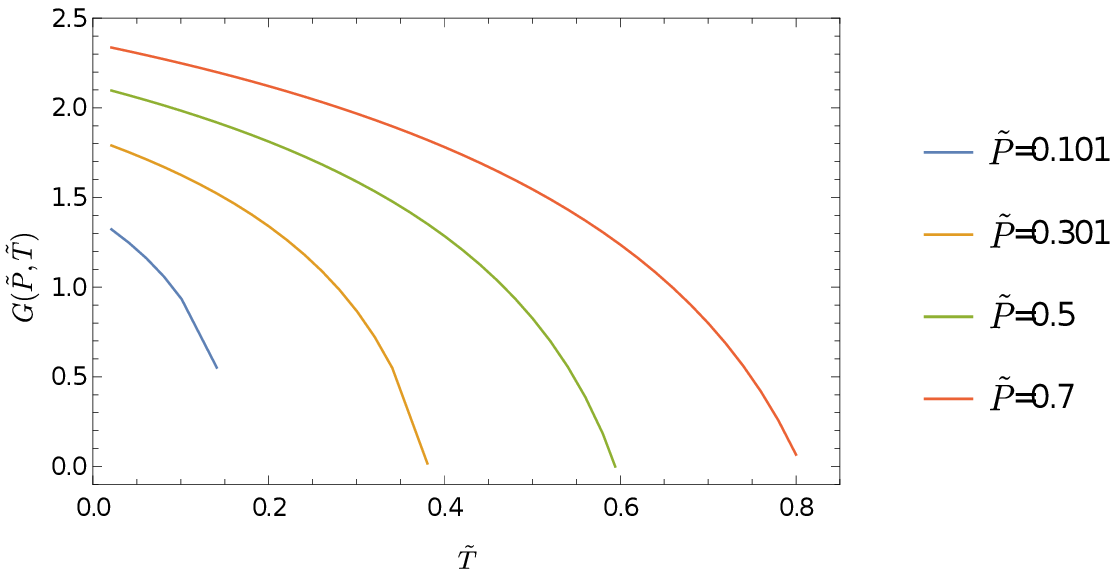}}
        \caption{Cuts of the Gibbs free energy at constant $\tilde{P}$.} \label{fig:Ptildeconst}
\end{figure}

We conclude that for the Kerr--AdS black hole the standard Maxwell 
equal area law is inconsistent with the GGD identity. Besides,
the use of the generalized intensive variables proposed here as the parameters defining thermodynamic equilibrium
seems to indicate that there is no first order phase transition between large and small black holes, as previously 
argued in the literature. However, our results deserve more investigation. Perhaps a comparison with explicit models 
directly constructed from statistical mechanics could shed more light on the validity of such statements. Alternatively, 
an analysis involving thermodynamic response functions could be interesting, although the significance of these 
response functions in the context of a generalized zeroth law should be re-evaluated.

\section{Conclusions and future directions}
\label{sec:conc}
In this work we consider a generalization of the zeroth law of thermodynamics for systems whose thermodynamic entropy 
is a quasi-homogeneous function of the (would-be) extensive variables (Def.\,\ref{def.TE}). Originating from the generalized 
version of the Gibbs-Duhem identity, we show how to define the generalized intensive variables that can be used to define 
thermodynamic equilibrium in such general cases (Def.\,\ref{def.GIV}). Moreover, we prove that this new definition 
resolves an inconsistency between the use of the standard Maxwell equal area law and the GGD identity 
that is usually overlooked, especially in the literature regarding the thermodynamics of black holes. Within this context, 
we consider two examples where the application of our generalized zeroth law should be relevant, namely the Schwarzschild 
and the Kerr--AdS black holes. 
The former is important because with our approach we recover a previous result found in~\cite{czinner2015black}, 
derived from a different perspective. The latter example is of interest because in the usual treatment the Kerr--AdS 
family of black holes shows a behavior which is very similar to that of a van der Waals fluid, including a first order 
phase transition. However, we argue that the use of the standard Maxwell equal area law in such case is not fully consistent and that
using the generalized intensive variables that we have introduced here in order 
to define thermodynamic equilibrium, such phase transition disappears. This statement however should be further 
investigated in other contexts in order to corroborate such a conclusion. 

Our results are intended to be a step forward towards a deeper 
formal understanding of the thermodynamic properties of systems 
with quasi-homogeneous entropy. However, it also calls for more 
detailed investigations. 
One can use the arguments given here to understand whether other 
reported first order phase transitions in black holes are consistent 
with their respective GGD identities or not (cf.~e.g.~\cite{kubizvnak2012p,spallucci2013maxwell,spallucci2013maxwell?s,dolan2014vacuum,lan2015note,wei2015clapeyron,wei2015insight,mo2015coexistence,kubizvnak2017black}).
It would also be interesting to study the implications of the 
present analysis for the conditions of equilibrium between black 
holes and heat reservoirs, e.g.~a Schwarzschild black hole in a hot 
flat space. Moreover, we would 
like to extend the comparison between our approach and the one presented in~\cite{biro2011zeroth,czinner2015black,czinner2016kerr,
biro2013q,czinner2015RenyiKerr} to other cases to see whether the agreement we found for the Schwarzschild black hole 
holds in more general contexts. 
Besides, it would be worth using explicit calculations as in~\cite{ramirez2008violation,
ramirez2008systems} to check whether our prediction of the new thermodynamic parameters defining equilibrium can be tested 
by numerical experiments, and to compare our results with the formalism proposed in~\cite{PhysRevE.88.042135,PhysRevLett.114.230601} 
presenting a different instance of a GGD relation for systems with long-range interactions. 
These directions will be the subject of future work.\\

\section*{Acknowledgements}
A.B. was supported by a DGAPA--UNAM postdoctoral fellowship. C.G. was supported by an UNAM postdoctoral fellowship program. 
F.N. received support from PAPIIT-UNAM Grant IN-111617.

\bibliographystyle{ieeetr}
\bibliography{references}

\begin{thebibliography}{10}

\bibitem{dolan2015black}
B.~P. Dolan, ``Black holes and {B}oyle's law--{T}he thermodynamics of the
  cosmological constant,'' {\em Modern Physics Letters A}, vol.~30, no.~03n04,
  p.~1540002, 2015.

\bibitem{oppenheim2003thermodynamics}
J.~Oppenheim, ``Thermodynamics with long-range interactions: From {I}sing
  models to black holes,'' {\em Physical Review E}, vol.~68, no.~1, p.~016108,
  2003.

\bibitem{abe2001general}
S.~Abe, ``General pseudoadditivity of composable entropy prescribed by the
  existence of equilibrium,'' {\em Physical Review E}, vol.~63, no.~6,
  p.~061105, 2001.

\bibitem{ramirez2008violation}
A.~Ram{\'\i}rez-Hern{\'a}ndez, H.~Larralde, and F.~Leyvraz, ``Violation of the
  zeroth law of thermodynamics in systems with negative specific heat,'' {\em
  Physical review letters}, vol.~100, no.~12, p.~120601, 2008.

\bibitem{ramirez2008systems}
A.~Ram{\'\i}rez-Hern{\'a}ndez, H.~Larralde, and F.~Leyvraz, ``Systems with
  negative specific heat in thermal contact: Violation of the zeroth law,''
  {\em Physical Review E}, vol.~78, no.~6, p.~061133, 2008.

\bibitem{biro2011zeroth}
T.~Bir{\'o} and P.~V{\'a}n, ``Zeroth law compatibility of nonadditive
  thermodynamics,'' {\em Physical Review E}, vol.~83, no.~6, p.~061147, 2011.

\bibitem{lenzi2012extensive}
E.~Lenzi and A.~Scarfone, ``Extensive-like and intensive-like thermodynamical
  variables in generalized thermostatistics,'' {\em Physica A: Statistical
  Mechanics and its Applications}, vol.~391, no.~8, pp.~2543--2555, 2012.

\bibitem{haggard2013death}
H.~M. Haggard and C.~Rovelli, ``Death and resurrection of the zeroth principle
  of thermodynamics,'' {\em Physical Review D}, vol.~87, no.~8, p.~084001,
  2013.

\bibitem{velazquez2016remarks}
L.~Velazquez, ``Remarks about the thermodynamics of astrophysical systems in
  mutual interaction and related notions,'' {\em Journal of Statistical
  Mechanics: Theory and Experiment}, vol.~2016, no.~3, p.~033105, 2016.

\bibitem{2012Dola}
B.~P. Dolan, ``Where is the {P}d{V} term in the fist law of black hole
  thermodynamics?,'' {\em arXiv preprint arXiv:1209.1272}, 2012.

\bibitem{kubizvnak2012p}
D.~Kubiz{\v{n}}{\'a}k and R.~B. Mann, ``P--{V} criticality of charged {A}d{S}
  black holes,'' {\em Journal of High Energy Physics}, vol.~2012, no.~7,
  pp.~1--25, 2012.

\bibitem{spallucci2013maxwell}
E.~Spallucci and A.~Smailagic, ``Maxwell's equal area law and the
  {H}awking-{P}age phase transition,'' {\em Journal of Gravity}, vol.~2013,
  2013.

\bibitem{spallucci2013maxwell?s}
E.~Spallucci and A.~Smailagic, ``Maxwell's equal-area law for charged {A}nti-de
  {S}itter black holes,'' {\em Physics Letters B}, vol.~723, no.~4,
  pp.~436--441, 2013.

\bibitem{dolan2014vacuum}
B.~P. Dolan, ``Vacuum energy and the latent heat of {A}d{S}-{K}err black
  holes,'' {\em Physical Review D}, vol.~90, no.~8, p.~084002, 2014.

\bibitem{lan2015note}
S.-Q. Lan, J.-X. Mo, and W.-B. Liu, ``A note on {M}axwellÕs equal area law for
  black hole phase transition,'' {\em The European Physical Journal C},
  vol.~75, no.~9, p.~419, 2015.

\bibitem{wei2015clapeyron}
S.-W. Wei and Y.-X. Liu, ``Clapeyron equations and fitting formula of the
  coexistence curve in the extended phase space of charged ads black holes,''
  {\em Physical Review D}, vol.~91, no.~4, p.~044018, 2015.

\bibitem{wei2015insight}
S.-W. Wei, Y.-X. Liu, {\em et~al.}, ``Insight into the microscopic structure of
  an ads black hole from a thermodynamical phase transition,'' {\em Physical
  review letters}, vol.~115, no.~11, p.~111302, 2015.

\bibitem{mo2015coexistence}
J.-X. Mo and G.-Q. Li, ``Coexistence curves and molecule number densities of
  ads black holes in the reduced parameter space,'' {\em Physical Review D},
  vol.~92, no.~2, p.~024055, 2015.

\bibitem{kubizvnak2017black}
D.~Kubiz{\v{n}}{\'a}k, R.~B. Mann, and M.~Teo, ``Black hole chemistry:
  thermodynamics with lambda,'' {\em Classical and Quantum Gravity}, vol.~34,
  no.~6, p.~063001, 2017.

\bibitem{callen2006thermodynamics}
H.~B. Callen, {\em Thermodynamics \& an Intro. to Thermostatistics}.
\newblock John wiley \& sons, 2006.

\bibitem{Bravetti2015377}
A.~Bravetti, C.~Lopez-Monsalvo, and F.~Nettel, ``Contact symmetries and
  {H}amiltonian thermodynamics,'' {\em Annals of Physics}, vol.~361, pp.~377 --
  400, 2015.

\bibitem{belgiorno2003quasi}
F.~Belgiorno, ``Quasi-homogeneous thermodynamics and black holes,'' {\em
  Journal of Mathematical Physics}, vol.~44, no.~3, pp.~1089--1128, 2003.

\bibitem{belgiorno2011general}
F.~Belgiorno and S.~Cacciatori, ``General symmetries: From homogeneous
  thermodynamics to black holes,'' {\em The European Physical Journal Plus},
  vol.~126, no.~9, pp.~1--19, 2011.

\bibitem{czinner2015black}
V.~G. Czinner, ``Black hole entropy and the zeroth law of thermodynamics,''
  {\em International Journal of Modern Physics D}, vol.~24, no.~09, p.~1542015,
  2015.

\bibitem{czinner2016kerr}
V.~Czinner and H.~Iguchi, ``A zeroth law compatible model to {K}err black hole
  thermodynamics,'' {\em Preprints 2016, 2016120096, doi:
  10.20944/preprints201612.0096.v1}, 2016.

\bibitem{biro2013q}
T.~S. Bir{\'o} and V.~G. Czinner, ``A q-parameter bound for particle spectra
  based on black hole thermodynamics with {R}{\'e}nyi entropy,'' {\em Physics
  Letters B}, vol.~726, no.~4, pp.~861--865, 2013.

\bibitem{czinner2015RenyiKerr}
V.~Czinner and H.~Iguchi, ``R\'enyi entropy and the thermodynamic stability of
  black holes,'' {\em Physics Letters B}, vol.~752, p.~306, 2016.

\bibitem{2009Kast}
D.~Kastor, S.~Ray, and J.~Traschen, ``Enthalpy and the mechanics of {A}d{S}
  black holes,'' {\em Classical and Quantum Gravity}, vol.~26, no.~19,
  p.~195011, 2009.

\bibitem{caldarelli2000thermodynamics}
M.~M. Caldarelli, G.~Cognola, and D.~Klemm, ``Thermodynamics of
  {K}err-{N}ewman-{A}d{S} black holes and conformal field theories,'' {\em
  Classical and Quantum Gravity}, vol.~17, no.~2, p.~399, 2000.

\bibitem{altamirano2014thermodynamics}
N.~Altamirano, D.~Kubiz{\v{n}}{\'a}k, R.~B. Mann, and Z.~Sherkatghanad,
  ``Thermodynamics of rotating black holes and black rings: phase transitions
  and thermodynamic volume,'' {\em Galaxies}, vol.~2, no.~1, pp.~89--159, 2014.

\bibitem{wei2016analytical}
S.-W. Wei, P.~Cheng, and Y.-X. Liu, ``Analytical and exact critical phenomena
  of d-dimensional singly spinning {K}err-{A}d{S} black holes,'' {\em Physical
  Review D}, vol.~93, no.~8, p.~084015, 2016.

\bibitem{PhysRevE.88.042135}
I.~Latella and A.~P\'erez-Madrid, ``Local thermodynamics and the generalized
  {G}ibbs-{D}uhem equation in systems with long-range interactions,'' {\em
  Phys. Rev. E}, vol.~88, p.~042135, Oct 2013.

\bibitem{PhysRevLett.114.230601}
I.~Latella, A.~P\'erez-Madrid, A.~Campa, L.~Casetti, and S.~Ruffo,
  ``Thermodynamics of nonadditive systems,'' {\em Phys. Rev. Lett.}, vol.~114,
  p.~230601, Jun 2015.

\end{thebibliography}

\end{document}